\documentclass{article}
\pdfoutput=1
\usepackage[utf8x]{inputenc}
\usepackage[english]{babel}
\usepackage[T1]{fontenc}
\usepackage[a4paper,top=2cm,bottom=2cm,left=3cm,right=3cm,marginparwidth=2.5cm]{geometry}
\usepackage{amsmath,amsthm}
\usepackage{graphicx}
\usepackage{float}
\usepackage{amsmath,amssymb,amsthm}
\usepackage{enumitem}
\usepackage{comment}
\usepackage{cleveref}

\newtheorem{lemma}{Lemma}

\newtheorem{theorem}{Theorem}

\theoremstyle{definition}
\newtheorem{example}{Example}

\newcommand{\zs}{\mathbb{Z}}

\begin{document}
\title{Linear difference 
operators with sequence coefficients having infinite-dimentional solution spaces}
\author{Sergei Abramov\\
\normalsize Federal Research Center ``Computer Science\\ 
\normalsize and Control'' of  
\normalsize the Russian Academy of Sciences,\\ 
\normalsize Moscow, Russia\\
\normalsize sergeyabramov@mail.ru \\
\and
Gleb Pogudin\\
\normalsize LIX, CNRS, \'Ecole polytechnique, \\
\normalsize Institute Polytechnique de Paris,\\
\normalsize  Paris, France \\	
\normalsize gleb.pogudin@polytechnique.edu}

\date{}
\maketitle

\rightline{\em In memory of Marko Petkov\v sek}

\begin{center}

    \vspace{1mm}
    \today
\end{center}
\begin{abstract}

The notion of lacunary infinite numerical sequence is introduced. It is shown that for an arbitrary
linear difference operator $L$ with coefficients belonging to the set $R$ of infinite numerical
sequences, a criterion (i.e., a necessary and sufficient condition) for the infinite-dimensionality of its space $V_L$
of solutions belonging to $R$ is the presence of a lacunary sequence in $V_L$.

\bigskip

\end{abstract}
\section{Introduction}
Finding sequences satisfying linear difference equations with constant coefficients (called C-finite sequences) is a classical and well-studied topic: the solution set forms a vector space and its dimension is equal to the order of the operator.
If one passes from linear to nonlinear equations, questions about the structure of the set of solutions in sequences become more complicated~\cite{nonlinear1} or even undecidable~\cite{nonlinear2}.
However, one can still define and study an interesting notion of the dimension of the solution set as was shown in~\cite{nonlinear3}.
In this paper, we take a different route by keeping the equations linear but substantially relaxing the restrictions on the coefficients: following~\cite{abp21}, we allow any sequences as coefficients.
In other words, we consider equations in an unknown sequence $\{x(n)\}_{n \in \mathbb{Z}}$:
\begin{equation}\label{eq:main_intro}
  a_r(n) x(n + r) + \ldots + a_1(n) x(n + 1) + a_0(n) x(n) = 0\quad \text{ for every }n\in \mathbb{Z},
\end{equation}
where $a_1, \ldots, a_n$ arbitrary sequences acting as coefficients.
One can see that the solutions of~\eqref{eq:main_intro} for a vector space.
The main question we study in this paper is \emph{under which conditions the solution space of~\eqref{eq:main_intro} has infinite dimension}.
In~\cite[Sec. 3]{abp21}, a specific sequence was exhibited such that, if it is a solution of the equation, then the solution space has infinite dimension.
We develop this approach to a complete criterion: our main result (Theorem~\ref{thm:main}) is that this happens if and only if~\eqref{eq:main_intro} has a solution of certain type which we call a ``lacunary sequences'' (by analogy with lacunry power series), that is, sequences containing arbitrary long finite zero intervals.

\section{Example of inifinite dimension}
The following example shows that indeed, if the coefficients of~\eqref{eq:main_intro} are arbitrary sequences, the dimension of the solution space can be infinite.

\begin{example}
\label{ex1}
Let $L=\Sigma _{k=0}^r c_k(n)\sigma^k$, where
\begin{equation}
\label{eq1}
(r+1) \nmid (n+k)\Rightarrow c_k(n)=0,
\end{equation}
and a sequence $a(n)$ be such that
\begin{equation}
\label{eq2}
(r+1) \mid n\Rightarrow a(n)=0,
\end{equation}
then by \eqref{eq1} we have
$c_k(n)a(n+k)=0$ for all such $n$, that $(r+1) \nmid (n+k)$.
In turn, \eqref{eq2} gives
$c_k(n)a(n+k)=0$ for all $n$ such that  $(r+1) \mid (n+k)$.
Thus
$$L(a(n))=\Sigma _{k=0}^rc_k(n)a(n+k)=0$$
for all  $n\in \zs$, i.e. $V_L$ contains all numeric sequences satisfying \eqref{eq2}.
As the values of $a(n)$ for $n$ such that  $(r+1) \nmid n$ can be chosen arbitrary, we conclude that
$\dim V_L =\infty$.
\end{example}

\section{Lacunary sequences}
We consider a linear difference equation in an unknown sequence
$\{x(n)\}_{n \in \mathbb{Z}}$:
\begin{equation}\label{eq:main}
a_r(n) x(n + r) + \ldots + a_1(n) x(n + 1) + a_0(n) x(n) = 0
\end{equation}
where $\{a_{0}(n)\}_{n \in \mathbb{Z}}, \ldots,\{a_{r}(n)\}_{n \in \mathbb{Z}}$ are arbitrary sequences.

We will call the number $r$ the \emph{order} of~\eqref{eq:main}.
For a sequence $\{a(n)\}_{n\in\mathbb{Z}}$, its \emph{support} is defined as
\[
\operatorname{supp}(\{a(n)\}) := \{i \in \mathbb{Z} \mid a(i) \neq 0\}.
\]
A sequence $\{a(n)\}_{n \in \mathbb{Z}}$ will be called \emph{lacunary} if the difference between the consequent elements of the support (that is, $i$ and $j$ in $\operatorname{supp}\{a(n)\}$ such that $i < j$ and there is no $k \in \operatorname{supp}\{a(n)\}$ with $i < k < j$) can be arbitrary large.

\begin{theorem}\label{thm:main}
The following statements are equivalent:
\begin{enumerate}
    \item\label{infdim} the dimension of the solution space of~\eqref{eq:main} is infinite;
    \item\label{lacunary} \eqref{eq:main} has a lacunary solution.
\end{enumerate}
\end{theorem}

The proof of the theorem will be based upon the following lemmas.

\begin{lemma}\label{lem:ray}
Assume that the dimension of the solution space of~\eqref{eq:main} is infinite.
Then there is a ray $\mathcal{S}$ in $\mathbb{Z}$, that is, a set of the form $\mathbb{Z}_{\geqslant i_0} := \{i \in \mathbb{Z} \mid i \geqslant i_0\}$ or $\mathbb{Z}_{\leqslant i_0} := \{i \in \mathbb{Z} \mid i \leqslant i_0\}$, such that the subspace of the solutions of~\eqref{eq:main} with the support contained in $\mathcal{S}$ is also infinite-dimensional.
\end{lemma}

\begin{proof}
Consider the subspace $V$ of the solution space of~\eqref{eq:main} consisting of the sequences $\{x(n)\}_{n \in \mathbb{Z}}$ satisfying $x(0) = x(1) = \ldots = x(r) = 0$.
This is a subspace of a finite codimension in the whole solution space (that is it is defined by finitely many constraints), so it must be also infinite-dimensional.
Consider a solution $\{x(n)\}$ of~\eqref{eq:main} belonging to $V$.
Since every next term in a solution of~\eqref{eq:main} depends only on the previous $r$, this solution can be represented as a sum of
\[
x^{+}(n) = \begin{cases}
 x(n), \text{ if } n > r,\\
 0, \text{ otherwise}
\end{cases}
\quad\text{and}\quad
x^{-}(n) = \begin{cases}
 x(n), \text{ if } n < 0,\\
 0, \text{ otherwise}
\end{cases}
\]
Thus, $V$ is a direct sum of subspaces of solutions of~\eqref{eq:main} with the support contained in $\mathbb{Z}_{> r}$ and $\mathbb{Z}_{< 0}$, respectively.
Therefore, at least one of these subspaces has infinite dimension.
\end{proof}

\begin{lemma}\label{lem:interval}
Assume that the dimension of the space $W$ of solutions of~\eqref{eq:main} with the support contained in $\mathbb{Z}_{\geqslant 0}$ is infinite.
There exists an index  $J > 0$ such that there exists a nonzero solution of~\eqref{eq:main} with the support contained in  $[1, \ldots, J]$.
\end{lemma}

\begin{proof}
For every $i > 0$, we denote by $W_i$ the projection of the space $W$ onto the coordinates $[1, \ldots, i]$.
Since the dimension of $W$ is infinite, the dimension of $W_i$ can be arbitrarily large.
Consider $J$ such that $\dim W_J \geqslant r + 2$.
We will prove that such $J$ satisfies the conditions of the lemma.
Consider a subspace of $W_J$ consisting of sequences satisfying
\[
x(J - r) = x(J - r + 1) = \ldots = x(J) = 0
\]
It has a codimension at most  $r + 1$, so its dimension is at least one.
Therefore, there exists a nonzero element in $W$ of the form
\[
x_0(n) = \begin{cases}
    0,\text{ if } n \leqslant 0,\\
    x(n),\text{ if } 0 < n < J - r,\\
    0, \text{ if } J - r \leqslant n \leqslant J,\\
    x(n), \text{ if } J < n.
\end{cases}
\]
Since the possibility of extending a finite solution to an infinite one depends only on the last $r$ terms, $W$ also must contain the following solution
\[
x_1(n) = \begin{cases}
    0,\text{ if } n \leqslant 0,\\
    x(n),\text{ if } 0 < n < J - r,\\
    0, \text{ if } J - r \leqslant n.
\end{cases}
\]
Since $\operatorname{supp}(x_1) \subset [1, \ldots, J]$, the lemma is proved.
\end{proof}

\begin{proof}[Proof of Theorem~\ref{thm:main}]
Implication $\ref{infdim} \implies \ref{lacunary}$.
Again, let $W$ be the space of solutions of~\eqref{eq:main} with the suppport belonging to $\mathbb{Z}_{\geqslant 0}$.
Lemma~\ref{lem:ray} implies that, after shifting and reversing indices if necessary, we can further assume that the dimension of $W$ is infinite.
We will inductively construct a sequence of indices $J_1, J_2,\ldots$ and solutions  $A_1, A_2, \ldots$ of~\eqref{eq:main} as follows.
We apply Lemma~\ref{lem:interval} and obtain index $J_1$ and the corresponding solution $A_1$.
Now we assume that the index  $J_\ell$ and solution $A_\ell$ are already constructed.
Then we apply Lemma~\ref{lem:interval} to the solutions of~\eqref{eq:main} with the support contained in $\mathbb{Z}_{> J_\ell + \ell}$.
We thus will obtain index $J_{\ell + 1}$ and solution $A_{\ell + 1}$.
The constructed solutions $A_1, A_2, \ldots$ have final supports, and the difference $\min\operatorname{supp}(A_{i + 1}) - \max\operatorname{supp}(A_i)$ is at least  $i + 1$ by construction.
Consider an infinite sum  $A_1 + A_2 + A_3 + \ldots$.
It is well-defined since the supports of the summands do not intersect.
This sum is the desired lacunary solution.

Implication $\ref{lacunary} \implies \ref{infdim}$.
Let $\{x(n)\}_{n \in \mathbb{Z}}$ be a lacunary solution of~\eqref{eq:main}.
Consider finite intervals of zeroes in $\{x(n)\}$ of length greater than $r$, that is, two indices $i < j$ such that 
\[
x(i) = x(i + 1) = \ldots = x(i + r) = x(j) = x(j + 1) = \ldots = x(j + r) = 0 \quad \text{ and }\quad \exists i < k < j\colon x(k) \neq 0.
\]
Since the order of the equation is $r$, the sequence
\[
x_0 = \begin{cases}
    x(n), \text{ if } i + r < n < j,\\
    0,\text{ otherwise}
\end{cases}
\]
is a nonzero solution of~\eqref{eq:main} with finite support. Repeating this operation for other pairs of zero intervals of length greater than $r$ which do not overlap with each other (this is always possible because there are infinitely many of them due to lacunarity), we obtain nonzero solutions $\{x_1(n)\}, \{x_2(n)\}, \ldots$ of~\eqref{eq:main} with finite non-intersecting supports. 
They are linearly independent, so the dimension of the whole solution space is infinite as well.
\end{proof}

\begin{example}
Go back to Example \ref{ex1}. The equation  $L(y)=0$ has lacunary solutions, e.g., the sequence 
\[
l(n) = \begin{cases}
 1, \text{ if } n=2^m(r+1)+1 \text{ for some } m\in \zs ,\\
 0, \text{ otherwise. }
\end{cases}\]
\end{example}

\end{document}